\documentclass[11pt]{article}
\usepackage{graphicx}
\usepackage{subcaption}
\usepackage{todonotes}
\usepackage{mathtools}
\usepackage{soul,comment,cite}
\usepackage{amsthm}
\usepackage{amsmath}
\usepackage{systeme}
\usepackage{amssymb}
\usepackage{gensymb}
\usepackage[ruled,vlined]{algorithm2e}
\SetKwComment{Comment}{/* }{ */}
\SetKwInOut{Input}{input}\SetKwInOut{Output}{output}
\newtheorem{theorem}{Theorem}

\newtheorem{lemma}{Lemma}
\newtheorem{observation}{Observation}

\newtheorem{definition}{Definition}

\graphicspath{{./figures/}}
\usepackage[absolute]{textpos}

\usepackage{fullpage}
\usepackage[top=0.8in, bottom=0.9in, left=0.9in, right=0.9in]{geometry}

\usepackage{latexsym}

\usepackage[pdftex, plainpages = false, pdfpagelabels, 
                 bookmarks=false,
                 bookmarksopen = true,
                 bookmarksnumbered = true,
                 breaklinks = true,
                 linktocpage,
                 pagebackref,
                 colorlinks = true,  
                 linkcolor = blue,
                 urlcolor  = cyan,
                 citecolor = red,
                 anchorcolor = green,
                 hyperindex = true,
                 hyperfigures
                 ]{hyperref}

\def\calC{\mathcal{C}}

\def\vl{\ell_{\perp}}

\title{Computing Maximum Cliques in Unit Disk Graphs\thanks{A preliminary version of this paper will appear in {\em Proceedings of the 37th Canadian Conference on Computational Geometry (CCCG 2025)}. This research was supported in part by NSF under Grant CCF-2300356.} 
}

\author{Anastasiia Tkachenko\thanks{Kahlert School of Computing,
University of Utah, Salt Lake City, UT 84112, USA. {\tt anastasiia.tkachenko@utah.edu}}
\and
Haitao Wang\thanks{Kahlert School of Computing,
University of Utah, Salt Lake City, UT 84112, USA. {\tt haitao.wang@utah.edu}}
}
\date{}

\begin{document}

\maketitle

\begin{abstract}
Given a set $P$ of $n$ points in the plane, the unit-disk graph $G(P)$ is a graph with $P$ as its vertex set such that two points of $P$ have an edge if their Euclidean distance is at most $1$. We consider the problem of computing a maximum clique in $G(P)$. The previously best algorithm for the problem runs in $O(n^{7/3+o(1)})$ time. 
We show that the problem can be solved in $O(n \log n + n K^{4/3+o(1)})$ time, where $K$ is the maximum clique size. The algorithm is faster than the previous one when $K=o(n)$. In addition, if $P$ is in convex position, we give a randomized algorithm that runs in $O(n^{15/7+o(1)})= O(n^{2.143})$ worst-case time and the algorithm can compute a maximum clique with high probability. For points in convex position, one special case  we solve is when a point in the maximum clique is given; we present an $O(n^2\log n)$ time (deterministic) algorithm for this special case. 
\end{abstract}


\section{Introduction}
\label{sec:intro}
A {\em clique} in an undirected graph is a subset of vertices such that every two vertices in the subset share an edge. 
Finding maximum cliques in graphs is a classical problem, which has applications across various domains, including bioinformatics~\cite{ref:Malod-DogninMa10}, scheduling~\cite{ref:WeideAn10}, coding theory~\cite{ref:EtzionGr98}, and social network analysis~\cite{ref:BalasundaramCl11} (see also the survey~\cite{ref:WuA15}). Beyond its practical significance, the maximum clique problem is also closely connected to a number of important combinatorial optimization problems such as clique partitioning~\cite{ref:WangSo07}, graph clustering~\cite{ref:SchaefferGr07}, and set packing~\cite{ref:WuAn12}. These problems can often be directly formulated as instances of the maximum clique problem or contain subproblems that require computing a maximum clique.

In general, the maximum clique problem is NP-hard for many classes of graphs. This intractability persists even in geometrically constrained settings, such as intersection graphs of ellipses and triangles~\cite{ref:AmbühlTh05}, hybrid models combining unit disks with axis-aligned rectangles~\cite{ref:BonnetMa20, ref:EspenantFi23}, and intersection graphs of rays~\cite{ref:CabelloTh13}. In each of these cases, computing a maximum clique remains as hard as in general graphs.

Nevertheless, for several classes of geometric intersection graphs, the maximum clique problem becomes tractable. Polynomial time algorithms are known, for instance, for intersection graphs of axis-aligned rectangles~\cite{ref:ImaiFi83}, circle graphs~\cite{ref:NashAn10}, circular-arc graphs~\cite{ref:ApostolicoFi87}, and unit-disk graphs~\cite{ref:ClarkUn90, ref:EppsteinIt94, ref:EspenantFi23, ref:KeilTh25}. In this paper, we focus on the maximum clique problem in unit-disk graphs. Unit-disk graphs arise naturally in modeling wireless communication networks, where devices with identical transmission ranges can interact if they are within range of each other~\cite{ref:ClarkUn90, ref:PerkinsHi94, ref:PerkinsAd99, ref:BalisterCo05}. Relatively simple geometric structure of unit-disk graphs, paired with real-world relevance, has made them a central object of study in computational geometry and related fields. 

Let $P$ be a set of $n$ points in the plane. The {\em unit-disk graph} of $P$, denoted $G(P)$, is defined as the graph with vertex set $P$ such that two points of $P$ have an edge if their Euclidean distance is at most 1. Equivalently, $G(P)$ is the intersection graph of congruent disks of radius $1/2$ centered at the points in $P$, that is, each disk defines a vertex and two vertices are connected if their corresponding disks intersect. A \textit{maximum clique} of $G(P)$ is thus a largest subset of points of $P$ in which every two points are of distance at most $1$. Given $P$, we consider the problem of computing a maximum clique of $G(P)$. 

\subsection{Previous work}

Clark, Colbourn, and Johnson~\cite{ref:ClarkUn90} proposed the first polynomial time algorithm to compute a maximum clique in a unit-disk graph. Their algorithm runs in $O(n^{4.5})$ time, based on the following idea. For each edge $(p,q)$ of $G(P)$ (i.e., $|pq|\leq 1$, where $|pq|$ is the distance of $p$ and $q$), let $L(p,q)$ denote the intersection of the disk centered at $p$ with radius $|pq|$ and the disk centered at $q$ with radius $|pq|$; $L(p,q)$ is called a {\em lens}. A key insight of \cite{ref:ClarkUn90} was that the subgraph $G_{L(p,q)}(P)$ induced by points of $P$ inside $L(p,q)$ is cobipartite (a graph is cobipartite if its complement graph is bipartite). Consequently, finding a maximum clique in $G_{L(p,q)}(P)$ reduces to finding a maximum independent set in its complement graph, which is bipartite. The latter problem in turn reduces to computing a maximum matching in the bipartite graph, which can be solved in $O(n^{2.5})$ time with $n$ as the number of vertices in the graph~\cite{ref:HopcroftAn73}. If $(p,q)$ is a farthest pair of points in a maximum matching $M$ of $G(P)$, then a crucial observation is that all points of $M$ must be in the lens $L(p,q)$. As such, by finding maximum cliques in the subgraphs for all lenses of all $O(n^2)$ of edges $G(P)$, one can compute a maximum clique in $G(P)$ in $O(n^{4.5})$ time.

Later, Aggarwal, Imai, Katoh, and Suri~\cite{ref:AggarwalFi91} gave an algorithm that can compute a maximum clique in a cobipartite subgraph of $G(P)$ (or equivalently compute a maximum independent set in its complement graph) in $O(n^{1.5}\log n)$ time by using the data structures in \cite{ref:HershbergerFi91,ref:ImaiEf86}. Plugging this algorithm into the above framework of \cite{ref:ClarkUn90} leads to an $O(n^{3.5}\log n)$ time algorithm for computing a maximum clique in $G(P)$. 

Eppstein and Erickson~\cite{ref:EppsteinIt94} subsequently derived an improved algorithm of $O(n^3 \log n)$ time by introducing an ordering of lenses and leveraging dynamic data structures. More specifically, the algorithm rotates a lens around each point in $P$, changing the corresponding cobipartite subgraph of the lens whenever a point enters or exits it.
At each change, instead of computing a maximum clique from scratch, the authors developed a dynamic data structure to maintain the maximum clique by searching an alternating path in the graph, so that each point update (insertion/deletion) to the subgraph can be handled in $O(n\log n)$ time.


Recently Espenant, Keil, and Mondal~\cite{ref:EspenantFi23} made a breakthrough and proposed the first subcubic time algorithm and their algorithm runs in $O(n^{2.5} \log n)$ time. It employs a divide-and-conquer approach. In the merge step, the authors discovered that it is sufficient to consider $O(n)$ lenses and compute a maximum clique for each of them, which takes $O(n^{2.5}\log n)$ time using the algorithm of \cite{ref:AggarwalFi91}. Hence, the merge step can be completed in $O(n^{2.5}\log n)$ time, so is the overall algorithm. 
As noted in \cite{ref:KeilTh25}, Timothy Chan observed that finding a maximum clique in a cobipartite graph can be solved faster in $O(n^{4/3+o(1)})$ time by decomposing $G(P)$ into bipartite cliques~\cite{ref:KatzAn97, ref:WangIm23} and applying the maximum matching algorithm~\cite{ref:FederCl95,ref:ChenMa22}. Plugging this new algorithm into the algorithm of \cite{ref:EspenantFi23} leads to an $O(n^{7/3 + o(1)})$ time algorithm to compute a maximum clique in $G(P)$.

\subsection{Our results}
In this paper, by using a grid and the above $O(n^{7/3 + o(1)})$ time algorithm, we show that computing a maximum clique in $G(P)$ can be done in  $O\bigl(n \log n + nK^{4/3 + o(1)}\bigr)$ time, where $K$ is the maximum clique size. If $K=o(n)$, then the runtime is asymptotically smaller than the previous algorithm. In particular, for sufficiently small $k$, e.g., $k=O(\sqrt{\log n})$, the runtime is bounded by $O(n\log n)$. 

We also consider a {\em convex position} case where points of $P$ are in a convex position, i.e., every point of $P$ is a vertex in the convex hull of $P$. In this case, we first show that if a point in a maximum clique of $P$ is provided, then we can compute a maximum clique in $O(n^2\log n)$ time. Combining this algorithm with the above $O\bigl(n \log n + nK^{4/3 + o(1)}\bigr)$ time algorithm, we present a randomized algorithm that can compute a clique in $G(P)$ in worst-case $O(n^{15/7+o(1)})= O(n^{2.143})$ time, and the clique is a maximum clique with high probability (i.e., with probability $1-1/n^c$ for an arbitrarily large constant $c$). 

Note that, although a constrained case, geometric problems for points in convex position have attracted much attention, e.g., Voronoi diagrams~\cite{ref:AggarwalA89}, $k$-center~\cite{ref:ChoiEf23,ref:TkachenkoDo25}, independent set~\cite{ref:SingireddyAl23,ref:TkachenkoDo25}, dispersion~\cite{ref:SingireddyAl23,ref:TkachenkoDo25}, dominating set~\cite{ref:TkachenkoDo25}, $\epsilon$-net~\cite{ref:ChazelleIm93}, triangulation~\cite{ref:LingasOn86}, Steiner tree~\cite{ref:RichardsA90}, etc. 
In particular, when the disks are in convex position (e.g., every disk appears in the convex hull of all disks; but the disks may have different radii), it has been reported that a maximum clique can be computed in polynomial time~\cite{ref:ÇağırıcıMa18}.

\paragraph{Outline.} The rest of the paper is organized as follows. After introducing some notation in Section~\ref{sec:pre}, we present our algorithm for the general case (i.e., points of $P$ are not necessarily in convex position) in Section~\ref{sec:kClique_GP}. In Section~\ref{sec:MaxCl_CP}, we describe the algorithm for the convex position case with the special assumption that a point in a maximum clique is provided. Section~\ref{sec:convexnopoint} finally solves the convex position case without the special assumption. 

\section{Preliminaries}
\label{sec:pre}

We introduce some notation that will be used throughout the paper, in addition to those already defined in Section~\ref{sec:intro}, such as $P$, $n$, and $G(P)$.

A \emph{unit disk} is any disk with radius $1$. For a point $p$ in the plane, we denote by $D_p$ the unit disk centered at $p$, and $\partial D_p$ the bounding circle of the disk. We use $x(p)$ and $y(p)$ to denote the $x$- and $y$-coordinates of $p$, respectively. 
Let $\ell_{\perp}(p)$ denote the vertical line through $p$.

For any two points $p$ and $q$ in the plane, let $|pq|$ represent their Euclidean distance, and $\overline{pq}$ the line segment with endpoints $p$ and $q$. 

For any subset $P'\subseteq P$, we let $G(P')$ denote the subgraph of $G(P)$ with vertex set $P'$. For simplicity, we call a subset $P'\subseteq P$ a {\em clique} if $G(P')$ is a complete graph, i.e., the distance between every two points of $P'$ is at most $1$. 


For any region $R$ in the plane, 
we use $P(R) = P \cap R$ to denote the subset of points from $P$ that lie in $R$. 



\section{The general case}
\label{sec:kClique_GP}
In this section, we present our algorithm for computing a maximum clique in the unit-disk graph $G(P)$, without assuming that $P$ is in convex position. We first solve a {\em decision problem}: Given a number $k$, determine whether $G(P)$ has a clique of size $k$ (and if so, report such a clique). We show that after $O(n\log n)$ time preprocessing, given any $k$, the decision problem can be solved in $O(nk^{4/3+o(1)})$ time. Later we will show that a maximum clique of $G(P)$ can be computed by an exponential search and using the decision algorithm. 

\subsection{The decision problem}

We first construct a set $\calC$ of grid cells to capture the proximity information for the points of $P$. 
The technique of using grids has been widely used in various algorithms for solving problems in unit-disk graphs~\cite{ref:WangRe23,ref:WangNe20, ref:ChanAl16,ref:WangUn23,ref:WangDy25}. 

The set $\calC$ has the following properties (see Figure~\ref{fig:C_N(C)}): 
(1) Each cell of $\calC$ is an axis-parallel square of side lengths $1/2$. This implies that the distance between every two points in each cell is less than $1$ and therefore points of $P$ in the cell form a clique. (2) The union of all cells of $\calC$ covers all the points of $P$. 
(3) Each cell $C\in \calC$ is associated with a subset $N(C)\subseteq \calC$ of $O(1)$ cells (called {\em neighboring cells} of $C$) such that for any point $q\in C$, $P(D_q) \subseteq \bigcup_{C'\in N(C)}P(C')$. 
(4) Each cell $C'\in \calC$ is in $N(C)$ for $O(1)$ cells $C\in \calC$. 

\begin{figure}
\begin{minipage}[t]{\textwidth}
\begin{center}
\includegraphics[height=2.0in]{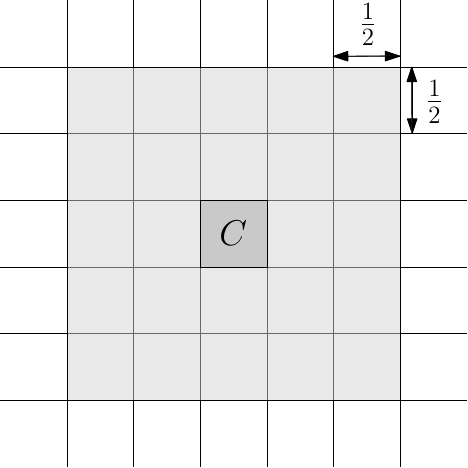}
\caption{The gray cells are neighboring cells of $C$.}
\label{fig:C_N(C)}
\end{center}
\end{minipage}
\end{figure}

An algorithm was provided in \cite{ref:WangUn23} for the following lemma. 

\begin{lemma}
\label{lem:grid}{\em \cite{ref:WangUn23}}
   The set $\calC$ of size $O(n)$, along with $P(C)$ and $N(C)$ for all cells $C \in \calC$, can be computed in $O(n\log n)$ time and $O(n)$ space.
\end{lemma}

Our preprocessing step is to construct the set $\calC$ by Lemma~\ref{lem:grid}. Given a number $k$, in the following we present our algorithm to determine whether $G(P)$ has a clique of size $k$. 

For any cell $C\in \calC$, if $|P(C)|\geq k$, we immediately find a clique of size $k$ as the points of $P(C)$ form a clique. 
In the following, we assume $|P(C)|<k$ for all cells $C\in \calC$. 

Suppose that $G(P)$ contains a clique $M$ of size $k$. Let $C$ be a cell that contains a point $p$ of $M$. Then, $M\subseteq P(D_p)$. By the third property of $\calC$, $P(D_p)\subseteq \bigcup_{C'\in N(C)}P(C')$. Hence, $M\subseteq \bigcup_{C'\in N(C)}P(C')$. Consequently, $M$ can be found in $G(P_C)$ for $P_C=\bigcup_{C'\in N(C)}P(C')$. Since $N(C)$ has $O(1)$ cells $C'$ and $|P(C')|<k$, we have $|P_C|=O(k)$. As such, applying the previous algorithm to $P_C$ can compute $M$ in $O(k^{7/3+o(1)})$ time. 

With the above observation, our algorithm works as follows. We compute $P_C$ for each cell $C\in \calC$. This takes $O(n)$ time in total as $|N(C)|\leq 1$ and each cell $C'\in \calC$ is in $N(C)$ for $O(1)$ cells $C\in \calC$ (the fourth property of $\calC$). If $|P_C|\geq k$, then we compute a maximum clique in $G(P_C)$, which takes $O(k^{7/3+o(1)})$ time. If the size of the computed clique is at least $k$, then we finish the algorithm and return the clique. The total time of the algorithm is $\sum_{C\in \calC, |P_C|\geq k} O(k^{7/3+o(1)})$, which is $O(nk^{4/3+o(1)})$ due to the following observation. 

\begin{observation}
    $\sum_{C\in \calC, |P_C|\geq k} k^{7/3+o(1)}=O(nk^{4/3+o(1)}).$
\end{observation}
\begin{proof}
Let $\calC'$ denote the subset of cells $C\in \calC$ with $|P_C|\geq k$. 
We claim that $|\calC'| = O(n/k)$. To see this, since each cell $C'\in \calC$ is in $N(C)$ for $O(1)$ cells of $C\in \calC$ (the fourth property of $\calC$), each point of $P$ is in $P_C$ for $O(1)$ cells of $C\in \calC$. Hence, $\sum_{C\in \calC'}|P_C|=O(n)$. Since $|P_C|\geq k$ for each cell $C\in \calC'$, we obtain $|\calC'|=O(n/k)$. This leads to the observation. 
\end{proof}

We thus have the following lemma. 

\begin{lemma}\label{lem:decision}
After $O(n\log n)$ time preprocessing, given any number $k$, we can decide whether $G(P)$ has a clique of size $k$ (and if so, report such a clique) in $O(nk^{4/3+o(1)})$ time. 
\end{lemma}

\subsection{The optimization problem: Finding a maximum clique}

Using Lemma~\ref{lem:decision}, the following theorem gives an algorithm to compute a maximum clique. 

\begin{theorem}\label{thm:k-clique}
Given a set $P$ of $n$ points in the plane, we can find a maximum clique in the unit-disk graph $G(P)$
in $O(n\log n+ nK^{4/3+o(1)})$ time, where $K$ is the maximum clique size.  
\end{theorem}
\begin{proof}
We first perform the preprocessing step of Lemma~\ref{lem:grid}, which takes $O(n\log n)$ time. Then, using the decision algorithm in Lemma~\ref{lem:decision}, we can compute a maximum clique by doing 
exponential search. Specifically, we call the decision algorithm for $k=1,2,4,8$ until a value $k$ such that the decision algorithm determines that $G(P)$ has a clique of size $k$ but does not have a clique of size $2k$. Let $k'=k$. Hence, we have $K\in [k',2k']$.  The above calls the decision algorithm $O(\log k')=O(\log K)$ times. We then do binary search with $k\in [k',2k']$ using the decision algorithm. This can compute a maximum clique by calling the decision algorithm another $O(\log K)$ times. As such, in total we call the decision algorithm $O(\log K)$ times, which takes  $O(nK^{4/3+o(1)}\log K)$. 

Therefore, the time complexity of the overall algorithm is $O(n\log n+ nK^{4/3+o(1)}\log K)$. Observe that the $\log K$ factor is subsumed by $K^{o(1)}$. The theorem thus follows. 
\end{proof}

\section{The convex position case with a given point}
\label{sec:MaxCl_CP}

In this section, we consider the convex position case with a special assumption that a point $p^*$ in a maximum clique of $G(P)$ is also provided. With $p^*$, we give an algorithm that can compute a maximum clique in $G(P)$ in $O(n^2\log n)$ time. We make a general position assumption that no three points of $P$ lie on the same line. 

\paragraph{Algorithm overview.}
Let $M^*$ denote a maximum clique of $G(P)$ containing $p^*$. Our algorithm will maintain a subset $S\subseteq P$ such that $G(S)$ is cobipartite (note that $G(S)$ is cobipartite if and only $S$ can be partitioned into two subsets $S_1$ and $S_2$ that are both cliques). During the course of the algorithm, there are $O(n)$ updates to $S$ and each update is either to insert a point of $P$ into $S$ or delete a point from $S$, and after each update $S$ is always cobipartite. We use the data structure of Eppstein and Erickson~\cite{ref:EppsteinIt94} (referred to as the {\em EE data structure}) to maintain a maximum clique of $S$ so that each  update can be handled in $O(n\log n)$ time (i.e., after each update, a maximum clique of $S$ can be computed in $O(n\log n)$ time). Since there are $O(n)$ updates to $S$, the total time of the algorithm is $O(n^2\log n)$. After all updates are done, among all cliques that have been computed, we return the largest one as our solution, i.e., a maximum clique for $G(P)$. To prove the correctness, we will show that at some moment during the algorithm $M^*\subseteq S$ holds (and therefore a maximum clique of that $S$ is also a maximum clique of $G(S)$). 

In what follows, we first define $S$ in Section~\ref{sec:define}, and then explain the updates to $S$ in Section~\ref{sec:update}. In fact, the exact definition of $S$ will not be clear until Section~\ref{sec:update}, but Section~\ref{sec:define} at least provides motivations to define $S$. Finally, we summarize the entire algorithm in Section~\ref{sec:imple}.

\subsection{Defining $\boldsymbol{S}$}
\label{sec:define}

Since $p^*\in M^*$, we have $M^*\subseteq P(D_{p^*})$. Hence, it suffices to consider the points of $P$ in the unit disk $D_{p^*}$. For notational convenience, we assume that all points of $P$ are in $D_{p^*}$, i.e., $|pp^*|\leq 1$ for all points $p\in P$. Also, we rotate the plane so that $p^*$ is the leftmost point and no two points of $P$ have the same $x$- or $y$-coordinate. 

Let $\hat{p}$ denote the rightmost point of $P$. Let $\hat{q}$ denote the rightmost point of $M^*$. Without loss of generality, we assume that $\hat{q}$ is on the upper hull of the convex hull of $P$ (otherwise, the following discussion is symmetric; as we do not know whether $\hat{q}$ is on the upper or lower hull of $P$, we actually run two versions of the algorithm and then return the better solution, i.e., the larger clique found by the two versions). 

Let $P_u$ denote the set of points of $P$ on the upper hull of $P$, and let $P_l=P\setminus P_u$. We assume that both $p^*$ and $\hat{p}$ are in $P_u$ but not in $P_l$. 

We order the points of $P_u$ from left to right as $p^*=p_1,p_2,\ldots, p_t=\hat{p}$. For each point $p_i\in P_u$, let $S_u(i)$ denote a set of points $p\in P_u$ such that $x(p)\leq x(p_i)$ and $|pp_i|\leq 1$; similarly, define $S_l(i)$ as the set of points of $p\in P_l$ such that $x(p)\leq x(p_i)$ and $|pp_i|\leq 1$. Let $S(i)=S_u(i)\cup S_l(i)$. 

Suppose $\hat{q}=p_i$ for some $i$; then since every point $p$ of $M^*$ satisfies $x(p)\leq x(p_i)$ and $|pp_i|\leq 1$, it holds that $M^*\subseteq S(i)$. Our algorithm is based on this observation. 

Our algorithm will scan the points of $P_u$ from left to right. For each point $p_i\in P_u$, we will maintain a subset $S=S'(i)$ of points of $P$ whose $x$-coordinates are less than or equal to $x(p_i)$ such that $S(i)\subseteq S'(i)$. We will compute a maximum clique for $G(S'(i))$ for each $i$, and finally return the largest clique found throughout the algorithm. This guarantees the correctness of the algorithm as $M^*\subseteq S(i)\subseteq S'(i)$. To compute maximum cliques of $G(S'(i))$ for all $1\leq i\leq t$, we will show that each $G(S'(i))$ is cobipartite. Furthermore, if $S=S'(i)$, we will show that through a sequence of updates (point insertions/deletions) to $S$, we can obtain $S=S'(i+1)$ and after each update $G(S)$ is always cobipartite. In this way, we can use the EE data structure to compute the maximum cliques of $G(S'(i))$ for all $1\leq i\leq t$. We will show that the total number of updates to $S$ in the whole algorithm is $O(n)$. Hence, the total time of the algorithm is $O(n^2\log n)$ since each update can be handled in $O(n\log n)$ time by the EE data structure. 

To show that $S'(i)$ is cobipartite, it suffices to show that $S'(i)$ can be partitioned into two subsets $S_1$ and $S_2$ that are both cliques. To this end, we start with the following observation, which essentially shows that $G(S(i))$ is cobipartite. 

\begin{observation}
\label{obs:up_clique}
Both $S_u(i)$ and $S_l(i)$ are cliques. 
\end{observation}
\begin{proof}
We only prove the case for $S_u(i)$ since the other case is symmetric. Consider any two points $p,p'\in S_u(i)$. It suffices to prove that $|pp'|\leq 1$.

By definition, both $p$ and $p'$ are to the left of $p_i$ and within distance 1 from both $p_1$ and $p_i$  (see Figure~\ref{fig:Cu_Cl_cliques}). Hence, $|pp'|\leq 1$ obviously holds if $\{p,p'\}\cap \{p^*,p_i\}\neq \emptyset$. In the following, we assume that $\{p,p'\}\cap \{p^*,p_i\}= \emptyset$. Without loss of generality, we assume that $x(p)\leq  x(p')$. 

\begin{figure}
\begin{minipage}[t]{\textwidth}
\begin{center}
\includegraphics[height=1.6in]{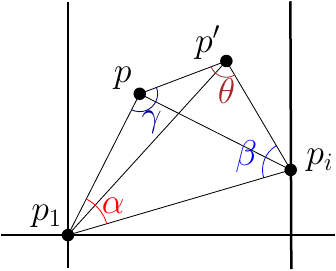}
\caption{Illustrating the proof of Observation~\ref{obs:up_clique}.}
\label{fig:Cu_Cl_cliques}
\end{center}
\end{minipage}
\end{figure}

Consider the two vertical lines $\ell_{\perp}(p_1)$ and $\ell_{\perp}(p_i)$. Since they are parallel, the sum of the interior angles at $p_1$ and $p_i$ in the quadrilateral $\Box p_1pp'p_i$ satisfies $\alpha + \beta \leq 180^\circ$, where $\alpha=\angle pp_1p_i$ and $\beta=\angle p'p_ip_1$. Therefore, the remaining two angles of the quadrilateral, denoted $\gamma=\angle p_1pp'$ and $\theta=pp'p_i$, must satisfy $\gamma + \theta \geq 180^\circ$. It follows that at least one of these angles, e.g., $\gamma$, is larger than or equal to $90^\circ$. Then in triangle $\triangle p_1pp'$, the side opposite this angle, namely $\overline{p_1p'}$, is no shorter than the side $\overline{pp'}$. Since $|p_1p'| \leq 1$, we conclude that $|pp'| \leq 1$.
\end{proof}

\paragraph{Remark.} To define $S'(i)$ so that $G(S'(i))$ is cobipartite, due to Observation~\ref{obs:up_clique}, it is tempted to have $S'(i)=S(i)$. However, as will be demonstrated later, having $S'(i)=S(i)$ may involve $\Omega(n^2)$ updates to $S'(i)$ for all $1\leq i\leq t$. Fortunately, we discover a way to make a larger $S'(i)$ so that $G(S'(i))$ is still cobipartite and $O(n)$ updates to $S'(i)$ are guaranteed. 
\bigskip



We define $S'(i)=S_u(i)\cup S'_l(i)$ for some subset $S'_l(i)$ such that $S_l(i)\subseteq S_l'(i)\subseteq P_l$. We will give the exact definition of $S'_l(i)$ later. By Observation~\ref{obs:up_clique}, $S_u(i)$ is a clique in $G(P)$, called an {\em upper clique}. As will be seen later, $S'_l(i)$ is also a clique, called a {\em lower clique}. 

\subsection{Updating $\boldsymbol{S=S'(i)}$} 
\label{sec:update}

We now explain how to update $S=S'(i)$ incrementally from $i=1$ to $i=t$. In the meanwhile, we will also give the exact definition of $S'_l(i)$.  

We say that a point $p\in P$ is {\em inserted} into $S_u(i)$ if $p\in S_u(i)$ but $p\not\in S_u(i-1)$, and $p$ is {\em deleted} from $S_u(i)$ if $p\in S_u(i-1)$ but $p\not\in S_u(i)$. Similarly, we use the same terminology for $S_l(i)$ (and also for $S'_l(i)$, $S(i)$ and $S'(i)$). 

Let $h$ be the index in $[1,t]$ such that $p_h$ is the point of $P_u$ with the largest $y$-coordinate.

\subsubsection{Maintaining the upper clique}
The following lemma implies that the number of updates to the upper clique $S_u(i)$ from $i=1$ to $i=t$ is $O(n)$. 

\begin{lemma}
\label{lem:upClique}
Any point $p\in P$ can be inserted into or deleted from $S_u(i)$ at most once for all $1\leq i\leq t$. In particular:
\begin{enumerate}
    \item If $p\in S_u(i)$ for some $i\in [1,h]$, then $p\in S_u(j)$ holds for all $j\in (i,h]$. 
    \item If $p$ is deleted from $S_u(i)$ for some $i\in [h,t]$, then $p\notin S_u(j)$ holds for all $j\in (i,t]$. 
\end{enumerate}
\end{lemma}
\begin{proof}
To prove the first part, by definition, it suffices to show that $|pp_j|\leq 1$. 
As $p\in S_u(i)$ and $i\in [1,h]$, $p$ is on the upper hull of $P$ to the left of $p_i$. Since $j\in (i,h]$, we have 
$x(p) <x(p_j)$, $y(p) < y(p_j)$, and the three points $p_1$, $p$, and $p_j$ appear in clockwise order along the upper hull of $P$  (see Figure~\ref{fig:upper10}). 
Hence, the angle $\angle p_1 p p_j$ must satisfy $90^\circ \leq \angle p_1 p p_j < 180^\circ$, forming an obtuse triangle. 
In such a triangle, the longest side is $\overline{p_1p_j}$. Since $|p_1p_j| \leq 1$, it follows that $|pp_j|\leq |p_1p_j| \leq 1$.

\begin{figure}
\begin{minipage}[t]{0.48\textwidth}
\begin{center}
\includegraphics[height=1.1in]{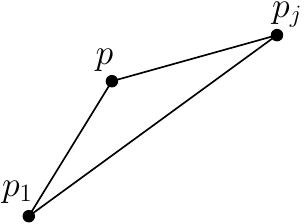}
\caption{Illustrating the proof of Lemma~\ref{lem:upClique}.}
\label{fig:upper10}
\end{center}
\end{minipage}
\hspace{0.02in}
\begin{minipage}[t]{0.48\textwidth}
\begin{center}
\includegraphics[height=1.1in]{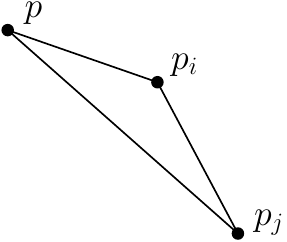}
\caption{Illustrating the proof of Lemma~\ref{lem:upClique}.}
\label{fig:upper20}
\end{center}
\end{minipage}
\end{figure}


For the second part of the lemma, since $p$ is deleted from $S_u(i)$, we have $x(p)<x(p_i)$ and $|pp_i| > 1$. To prove $p\not\in S_u(j)$, it suffices to show that $|pp_j|>1$. 

We first argue that $y(p) > y(p_i)$. Assume to the contrary that $y(p) \leq  y(p_i)$. Then, following a similar argument as above, since $p_1$, $p$, and $p_i$ appear in clockwise order along the upper hull of $P$, the angle $\angle p_1 p p_i$ must satisfy $90^\circ \leq \angle p_1 p p_i < 180^\circ$, forming an obtuse triangle (see Figure~\ref{fig:upper10} by changing $p_j$ to $p_i$). Hence, $|pp_i|\leq |p_1p_i| \leq 1$, contradicting $|pp_i| > 1$. Therefore, $y(p) > y(p_i)$ holds. 

Since $j\in (i,t]$ and $h\leq i$, we have $x(p_i)< x(p_j)$ and $y(p_i)> y(p_j)$ (see Figure~\ref{fig:upper20}). As $y(p) > y(p_i)> y(p_j)$, $x(p)<x(p_i)<x(p_j)$, and $p\rightarrow p_i \rightarrow p_j$ makes a right turn, the angle $\angle p p_i p_j$ must satisfy $90^\circ \leq \angle p p_i p_j < 180^\circ$, forming an obtuse triangle. Hence, $|pp_j|\geq |pp_i|$. As $|pp_i| > 1$, we obtain $|pp_j|>1$.

This completes the proof of the lemma.
\end{proof}


\subsubsection{Maintaining the lower clique}

To define the lower clique $S_l'(i)$, which is based on $S_l(i)$, and maintain it for updates, we first make the following observation, similar to Lemma~\ref{lem:upClique}(2) but for $i\in [1,h]$. The lemma implies that the number of updates to $S_l(i)$ for all $1\leq i\leq h$ is $O(n)$.

\begin{lemma}\label{lem:lowerfirsthalf}
Any point $p\in P$ can be deleted from $S_l(i)$ at most once for all $1\leq i\leq h$. In particular, if $p$ is deleted from $S_l(i)$ for some $i\in [1,h]$, then $p\notin S_l(j)$ holds for all $j\in (i,h]$. 
\end{lemma}
\begin{proof}
Since $p$ is deleted from $S_l(i)$, we have $|pp_i| > 1$ and $x(p) < x(p_i)$. To prove $p\notin S_l(j)$, it suffices to show that $|pp_j|>1$. 

Since $j\in (i,h]$, we have $x(p_i)< x(p_j)$ and $y(p_i)< y(p_j)$. As $x(p)< x(p_i)<x(p_j)$ and $p\rightarrow p_i \rightarrow p_j$ makes a right turn, the angle $\angle pp_ip_j$ must satisfy $90^\circ\leq\angle pp_ip_j <180\degree$. Hence, $|pp_j|\geq |pp_i|$. Since $|pp_i|>1$, we obtain $|pp_j|>1$.
\end{proof}

The above lemma is for indices $i\in [1,h]$. When $i$ changes from $h$ to $t$, however, the situation becomes more complicated since a point of $P_l$ may be inserted into and deleted from $S_l(i)$ up to $\Omega(n)$ times, leading to a total of $\Omega(n^2)$ updates to $S_l(i)$ in the worst case. Consider the example illustrated in Figure~\ref{fig:ex_O(n)ops}. Let $p$ be a point in $P_l$. Suppose $P_u$ contains $O(n)$ points such that every other point from left to right lies at distance exactly $1 - \delta$ from $p$, while its immediate neighbors lie at distance $1 + \delta$, for some small $\delta > 0$. In this setting, when $i$ changes from $h$ to $t$, $p$ must alternate between being inserted into and deleted from $S_l(i)$. Furthermore, if there are $O(n)$ such points $p$ in $P_l$ (e.g., all lying within distance $\delta$ from $p$), then each of them will behave similarly. In this case, the total number of updates to $S_l(i)$ for $h\leq i\leq t$ may reach $\Omega(n^2)$. In the following, we consider a new strategy to define and maintain a larger set $S_l'(i)$ with $S_l(i)\subseteq S_l'(i)$.



\begin{figure}[t]
\begin{minipage}[t]{0.48\textwidth}
\begin{center}
\includegraphics[height=2.0in]{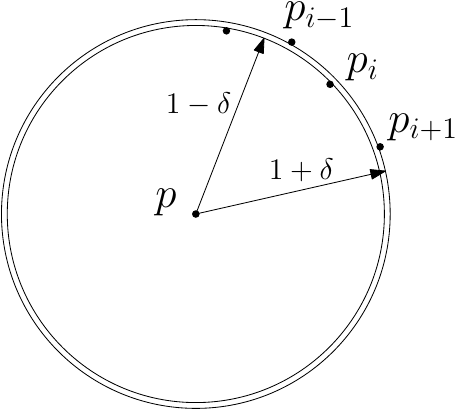}
\caption{Illustration of a point $p\in P_l$ that may undergo $\Omega(n)$ updates to $S_l(i)$.}
\label{fig:ex_O(n)ops}
\end{center}
\end{minipage}
\hspace{0.15in}
\begin{minipage}[t]{0.48\textwidth}
\begin{center}
\includegraphics[height=2.0in]{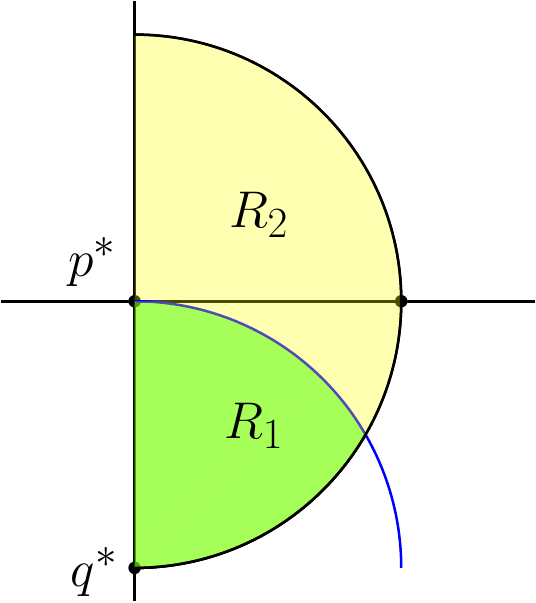}
\caption{Illustrating the two regions $R_1$ and $R_2$ in the right half of $D_{p^*}$.}
\label{fig:CP_clique_search}
\end{center}
\end{minipage}
\end{figure}

Let $q^*$ be the point $(x(p^*),y(p^*)-1)$. We partition the right half of the unit disk $D_{p^*}$ into two regions $R_1$ and $R_2$ by the boundary of $D_{q^*}$ (see Figure~\ref{fig:CP_clique_search}). Specifically, define $R_1 =\{p\in D_{p^*}: p\in D_{q^*}, x(p)\geq x(p^*)\}$ and $R_2 =\{p\in D_{p^*}: p\notin D_{q^*}, x(p) \geq x(p^*)\}$.
We first have the following lemma for $R_2$.

\begin{lemma}
\label{obs:IIpoints}
For any point $p$ of $P_l$ in $R_2$, $|pp_i|\leq 1$ holds for any $p_i\in P_u$ with $x(p)<x(p_i)$. This implies that if $p$ is inserted into $S_l(i)$ for some $i\in [h,t]$, then $p$ will always be in $S_l(j)$ for any $j\in (i,t]$. 
\end{lemma}
\begin{proof}
It suffices to prove that $|pp_i|\leq 1$. This obviously holds if $p=p^*$. Hence, we assume $p\neq p^*$ below, and thus $x(p^*)<x(p)$. Depending on whether $y(p)\geq y(p^*)$, there are two cases. 

We first prove the case $y(p)\geq y(p^*)$ since it is easier. In this case, since $p\in S_l$ while $p_i\in S_u$, $p^*\rightarrow p\rightarrow p_i$ makes a left turn. Since $x(p)<x(p_i)$, we have $\angle p^*pp_i\geq 90^\circ$ (see Figure~\ref{fig:lower10}). Therefore, $|pp_i|\leq |p^*p_i|$. As $|p^*p_i|\leq 1$, we obtain $|pp_i|\leq 1$.

In the following, we prove the case $y(p)< y(p^*)$; refer to Figure~\ref{fig:Region_III}. Our goal is to prove $|pp_i|\leq 1$. 

Let $\ell_1$ denote the bisector of $p^*$ and $p$. Since $y(p)< y(p^*)$ and $x(p)>x(p^*)$, $p$ is in the fourth quadrant of $p^*$ and the slope of $\ell_1$ is positive (see Figure~\ref{fig:Region_III}). 
We claim that $x(q_1)\leq x(p)$. Before proving the claim, we use the claim to prove $|pp_i|\leq 1$. Indeed, since $x(q_1)\leq x(p)<x(p_i)$, $p_i\in D_{p^*}$, and the slope of $\ell_1$ is positive, $p_i$ must lie to the right of $\ell_1$. This means that $|pp_i|\leq |p^*p_i|$. As $|p^*p_i|\leq 1$, we obtain $|pp_i|\leq 1$. 
In the following, we prove the claim $x(q_1)\leq x(p)$.

\begin{figure}
\begin{minipage}[t]{0.48\textwidth}
\begin{center}
\includegraphics[height=1.1in]{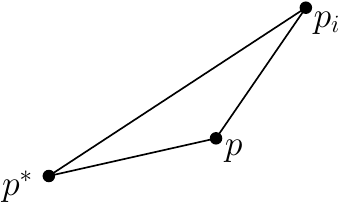}
\caption{Illustrating the case $y(p)\geq y(p^*)$ for the proof of Lemma~\ref{obs:IIpoints}.}
\label{fig:lower10}
\end{center}
\end{minipage}
\hspace{0.05in}
\begin{minipage}[t]{0.48\textwidth}
\begin{center}
\includegraphics[height=2.3in]
{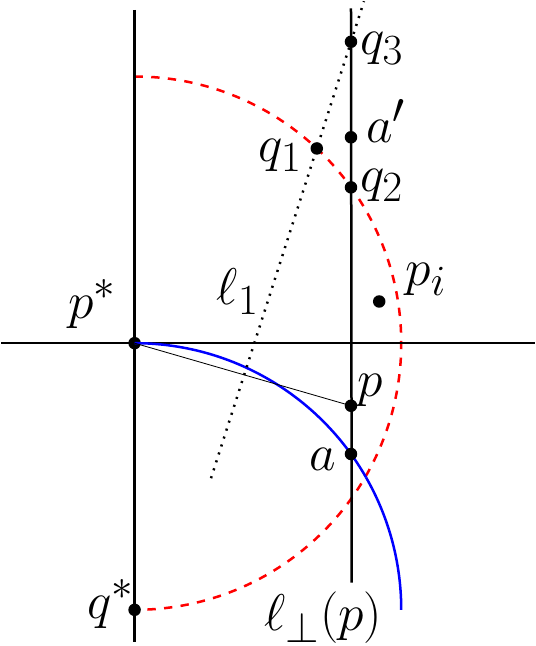}
\caption{Illustrating the case $y(p)< y(p^*)$ for the proof of Lemma~\ref{obs:IIpoints}: The red dashed arc is the right half boundary of $D_{p^*}$ and the blue solid arc is on the boundary of $D_{q^*}$.}
\label{fig:Region_III}
\end{center}
\end{minipage}
\end{figure}

The vertical line $\vl(p)$ intersects $\partial D_{p^*}$ at two points; we let $q_2$ denote the upper intersection point, which must be in the first quadrant of $p^*$. To prove $x(q_1)\leq x(p)$, since $x(p)=x(q_2)$, it suffices to prove $x(q_1)\leq x(q_2)$. Let $q_3$ be the intersection of $\ell_1$ and $\vl(p)$. To prove $x(q_1)\leq x(q_2)$, 
since slope of $\ell_1$ is positive, it suffices to prove that $y(q_3)\geq y(q_2)$. 

Since $p\in R_2$, the vertical line $\vl(p)$ intersects the upper boundary of the unit disk $D_{q^*}$ at a point $a$ with $y(a)< y(p)$. Since $q_2$ is at the upper boundary of $D_{p^*}$, $a$ is at the upper boundary of $D_{q^*}$, $p^*$ and $q^*$ are at the same vertical line with $y(p^*)=y(q^*)+1$, and $q_2$ and $a$ are at the same vertical line, we obtain that $y(q_2)=y(a)+1$. 

Let $a'$ be a point on $\vl(p)$ with $y(a')=y(p)+1$. Since $y(p) > y(a)$, we have $y(a')=y(p)+1> y(a)+1=y(q_2)$. This implies that $a'$ is outside $D_{p^*}$ and therefore $|p^*a'|>1=|pa'|$. Hence, $a'$ is strictly to the right of the bisector line $\ell_1$. Since the slope of $\ell_1$ is positive and $q_3=\ell_1\cap \vl(p)$, all points of $\vl(p)$ strictly to the right of $\ell_1$ are below $q_3$. Therefore, we obtain $y(q_3)\geq y(a')$. As $y(a')>y(q_2)$, we obtain $y(q_3)>y(q_2)$. 

This completes the proof of the lemma.
\end{proof}

\paragraph{Defining $\boldsymbol{S'_l(i)}$.}
In light of Lemma~\ref{obs:IIpoints}, when $i$ changes from $h$ to $t$, instead of maintaining $S_l(i)$, we maintain a superset $S'_l(i)$ of $S_l(i)$, defined as follows. For $1\leq i\leq h$, we simply define $S'_l(i)=S_l(i)$. In general, assuming that $S'_l(i)$ has been defined for some $i\in [h,t]$, we define $S'_l(i+1)$ as follows. First, we put all points of $S'_l(i)$ to $S'_l(i+1)$. Then, for each point $p\in S_l(i+1)\setminus S_l'(i)$, we add $p$ to $S'_l(i+1)$. Furthermore, if $p\in R_2$, then we check every point $p'$ from $S_l'(i+1)\cap R_1$, and delete $p'$ from $S'_l(i+1)$ if $|pp'|>1$ (note that if $|pp'|>1$, then $p'$ cannot be in $S_l(i+1)$ since $S_l(i+1)$ is a clique by Observation~\ref{obs:up_clique}; this means that no point $p\in S_l(i+1)\setminus S_l'(i)$ added to $S_l'(i+1)$ above will be later deleted; this further implies that the order we add points of $p\in S_l(i+1)\setminus S_l'(i)$ to $S_l'(i+1)$ does not matter). 

The following lemma shows that $S'_l(i)$ is a ``viable" alternative to $S_l(i)$. 

\begin{lemma}\label{lem:lowerclique}
For any $i\in [1,t]$, $S_l(i)\subseteq S_l'(i)$ and $S_l'(i)$ is a clique. 
\end{lemma}
\begin{proof}
First of all, by definition,  $S_l'(i)=S_l(i)$ for all $i\in [1,h]$. Therefore, the lemma statement obviously holds for all $i\in [1,h]$. We now prove the lemma for $i\in [h,t]$ by induction. Assuming that $S_l(i)\subseteq S_l'(i)$ and $S_l'(i)$ is a clique for some $i\in [h,t]$ (which is true when $i=h$), we prove below that $S_l(i+1)\subseteq S_l'(i+1)$ and $S_l'(i+1)$ is a clique. 

Let $S''=S_l(i+1)\cup S_l'(i)$. Let $S'$ be the set of points of $p'\in S''$ such that $p'\in R_1$ and $|p'p|>1$ for some point $p\in S_l(i+1)\setminus S_l'(i)$. By definition, $S_l'(i+1)=S''\setminus S'$. 

Clearly, $S_l(i+1)\subseteq S''$. We claim that $S_l(i+1)\cap S'=\emptyset$. Indeed, consider a point $p'\in S'$. By definition of $S'$, there exists a point $p\in S_l(i+1)$ with $|pp'|>1$. Since $p\in S_l(i+1)$ and $S_l(i+1)$ is a clique, $p'$ cannot be in $S_l(i+1)$. Therefore, $S_l(i+1)\cap S'=\emptyset$. As $S_l'(i+1)=S''\setminus S'$, we obtain that $S_l(i+1)\subseteq S_l'(i+1)$. 

To show that $S_l'(i+1)$ is a clique, recall that to define $S_l'(i+1)$, we first set $S_l'(i+1)=S'_l(i)$, which is a clique by definition. Then, whenever we add a point $p\in S_l(i+1)$ to $S_l'(i+1)$, we delete all points $p'$ in the current $S'_l(i+1)$ with $p'\in R_1$ and $|pp'|>1$. Also notice that for any point $p'$ in the current $S'_l(i+1)$ with $p'\in R_2$, 
by Lemma~\ref{obs:IIpoints}, $p'\in S_l(i+1)$ and thus $|pp'|\leq 1$ since $p\in S_l(i+1)$ and $S_l(i+1)$ is a clique. Therefore, after $p$ is inserted into $S_l'(i+1)$ and points $p'\in S_l'(i+1)$ with $|pp'|>1$ are deleted, $S'_l(i+1)$ is still a clique. This proves that $S_l'(i+1)$ is a clique. 
\end{proof}

The following lemma explains why $S_l'(i)$ is a better alternative to $S_l(i)$.

\begin{lemma}
\label{lem:relaxed_loClique}
    Any point $p\in P$ can be inserted into or deleted from $S'_l(i)$ at most $O(1)$ times for all $1\leq i\leq t$.
\end{lemma}
\begin{proof}
First of all, by Lemma~\ref{lem:lowerfirsthalf}, any point $p\in P$ can be deleted from $S'_l(i)=S_l(i)$ at most once for all $i\in [1,h]$, implying that $p$ can be inserted into $S'_l(i)$ at most once for all $i\in [1,h]$.

For $i\in (h,t]$, by our definition, whenever a point $p'$ is deleted from $S'_l(i)$, there must be a point $p\in S'_l(i)$ such that $p\in R_2$ and $|pp'|>1$. Since $p\in R_2$, $p$ will always be in $S'_l(j)$ for all $j\in (i,t]$ by Lemma~\ref{obs:IIpoints}. This means that $p'$ can never be inserted into $S'_l(j)$ again for all $j\in (i,t]$ since $S_l'(j)$ is a clique by Lemma~\ref{lem:lowerclique}. Therefore, each point will be deleted from $S'_l(i)$ at most once for all $i\in [h,t]$, implying that each point can be inserted into $S'_l(i)$ at most once for all $i\in (h,t]$. 

Combining the above, each point of $P$ can be inserted into or deleted from $S'_l(i)$ at most $O(1)$ times for all $1\leq i\leq t$.
\end{proof}

\begin{definition}
   Define $S'(i)$ to be $S_u(i)\cup S_l'(i)$ for all $1\leq i\leq t$.
\end{definition}

Lemmas~\ref{lem:upClique} and \ref{lem:relaxed_loClique} together lead to the following. 

\begin{lemma}\label{lem:combine}
    Any point $p\in P$ can be inserted into or deleted from $S'(i)$ at most $O(1)$ times for all $1\leq i\leq t$.
\end{lemma}

\subsection{Algorithm implementation and time analysis}
\label{sec:imple}

Based on the above discussions, we implement our algorithm as follows. Initially when $i=1$, we compute $S'(i)=S(i)=S_l(i)\cup S_u(i)$, which consists of a single point $p_1$. Therefore, $\{p_1\}$ is a maximum clique of $G(S'(i))$. In general, assume that we have computed $S'(i)$, whose graph $G(S'(i))$ is cobipartite, and its maximum clique $M(i)$. Then, we compute $S'(i+1)$ and its maximum clique $M(i+1)$ of $S'(i+1)$ as follows. Depending on whether $i+1\leq h$, there are two cases. 

\paragraph{The case $\boldsymbol{i+1\leq h}$.}
If $i+1\leq h$, then we first compute $S'(i+1)=S(i+1)=S_l(i+1)\cup S_u(i+1)$. Note that $S_l(i+1)$ and $S_u(i+1)$  can be easily computed in $O(n)$ time by checking every point of $P$. We then compute the set $I(i+1)=S'(i+1)\setminus S'(i)$ and the set $D(i+1)=S'(i)\setminus S'(i+1)$. Notice that in order to obtain $S'(i+1)$ from $S'(i)$, $I(i+1)$ is the set of points to be inserted into $S'(i+1)$ while $D(i+1)$ is the set of points to be deleted from $S'(i)$. Given $S'(i)$ and $S'(i+1)$, the two sets $I(i+1)$ and $D(i+1)$ can be easily computed in $O(n\log n)$ time. Let $S=S'(i)$. Recall that we already know a maximum clique of $S$. To compute a maximum clique of $S'(i+1)$, we use the EE data structure~\cite{ref:EppsteinIt94} to maintain the maximum clique of $S$, and perform the following sequence of updates to $S$: First delete the points of $D(i+1)$ from $S$ one by one in any arbitrary order and then insert the points of $I(i+1)$ into $S$ one by one in any arbitrary order. Observe that at any moment of the sequence of the updates, $S$ is either a subset of $S'(i)$ or a subset of $S'(i+1)$. As both $G(S'(i))$ and $G(S'(i+1))$ are cobipartite, $G(S)$ is always cobipartite during the above sequence of updates. Hence, the EE data structure can compute a maximum clique after each update in $O(n\log n)$ time. As such, the total time for computing a maximum clique $M(i+1)$ of $S'(i+1)$ is $O((|I(i+1)|+|D(i+1)|+1)\cdot n\log n)$. Note that the ``$+1$'' is because we spend $O(n\log n)$ time computing $S'(i+1)$, $I(i+1)$, and $D(i+1)$. 

\paragraph{The case $\boldsymbol{i+1> h}$.}
If $i+1> h$, then we have $S'(i)=S_u(i)\cup S_l'(i)$ and $S'(i+1)=S_u(i+1)\cup S_l'(i+1)$. We first compute $S_u(i+1)$. Then, we compute $I_u(i+1)=S_u(i+1)\setminus S_u(i)$ and $D_u(i+1)=S_u(i)\setminus S_u(i+1)$. 
Next, we compute $S_l(i+1)$ and $I_l(i+1)=S_l(i+1)\setminus S_l'(i)$. Subsequently, for each point $p\in I_l(i+1)$, we check every point $p'\in S'_l(i)\cap R_1$ and add $p'$ to $D_l(i+1)$ if $|pp'|>1$. By definition, $S'(i+1)$ can be obtained from $S'(i)$ by deleting the points of $D(i+1)=D_l(i+1)\cup D_u(i+1)$ and inserting the points of $I(i+1)=I_l(i+1)\cup I_u(i+1)$. 

Note that computing $S_u(i+1)$, $S_l(i+1)$, $I_u(i+1)$, $D_u(i+1)$, and $I_l(i+1)$ can be easily done in $O(n\log n)$ time. For each point $p\in I_l(i+1)$, we can spend $O(n)$ time checking every point $p'\in S_l'(i)$ to determine whether $p'$ should be added to $D_l(i+1)$. Hence, computing $D_l(i+1)$ can be done in $O(|I_l(i+1)|\cdot n)$ time.

As above, let $S=S'(i)$. Since we already know a maximum clique of $S$, to compute a maximum clique of $S'(i+1)$, we use the EE data structure to maintain the maximum clique of $S$, and perform the following sequence of updates to $S$: First delete the points of $D(i+1)$ from $S$ and then insert the points of $I(i+1)$ into $S$. 
As in the above case, $G(S)$ is always cobipartite after each update. Hence, the EE data structure can compute a maximum clique after each update in $O(n\log n)$ time. As such, the total time for computing a maximum clique $M(i+1)$ of $S'(i+1)$ is $O((|I(i+1)|+|D(i+1)|+1)\cdot n\log n)$. 
\bigskip


To summarize, in either case, computing a maximum clique $M(i+1)$ for each $S'(i+1)$ takes $O((|I(i+1)|+|D(i+1)|+1)\cdot n\log n)$ time. By Lemma~\ref{lem:combine}, we have $\sum_i(|I(i+1)|+|D(i+1)|)=O(n)$. Hence, the total time of the algorithm is $O(n^2\log n)$. We thus obtain the following theorem. 

\begin{theorem}\label{thm:convexgivenpoint}
Given a set $P$ of $n$ points in convex position in the plane, if a point in a maximum clique of $G(P)$ is provided, then a maximum clique of $G(P)$ can be computed in $O(n^2 \log n)$ time.
\end{theorem}

Note that for any point $p^*\in P$, our algorithm guarantees to find a clique no smaller than any clique containing $p^*$. Indeed, suppose that $M$ is a clique of $G(P)$ that contains $p^*$. Then, our algorithm guarantees $M\subseteq S'(i)$ for some $i\in [1,t]$. Since the algorithm finds a maximum clique for $S'(i)$ for each $i\in [1,t]$, it can always find a clique of size at least $|M|$. Therefore, we have the following more general result. 

\begin{theorem}\label{thm:convexgivenpointgeneral}
Let $P$ be a set of $n$ points in convex position in the plane and $p^*$ is a point of $P$. Suppose that $k$ is the size of a clique of $G(P)$ that contains $p^*$. Then, we can compute a clique whose size is at least $k$ in $O(n^2 \log n)$ time.
\end{theorem}

\section{The convex position case without a given point}
\label{sec:convexnopoint}

We now consider convex position case problem without assuming a point in a maximum clique is given. We present a randomized algorithm for the problem by combining the results in Theorems~\ref{thm:k-clique} and \ref{thm:convexgivenpoint}. 


Let $K$ denote the maximum clique size of $G(P)$. We first apply the algorithm of Lemma~\ref{lem:decision} with $k=n^{6/7}$, which can determine whether $K<n^{6/7}$ in $O(n^{15/7+o(1)})$ time. Depending on whether $K < n^{6/7}$, we run two different algorithms. 

If $K < n^{6/7}$, then we run the algorithm of Theorem~\ref{thm:k-clique}, which can find a maximum clique in $O(n\log n+ nK^{4/3+o(1)})$ time. The time complexity is thus bounded by $O(n^{15/7+o(1)})$. 

If $K\geq n^{6/7}$, then we randomly pick a point $p^*$ from $P$ and run the algorithm in Theorem~\ref{thm:convexgivenpoint}. Let $M^*$ be a maximum clique of $G(P)$. 
As $K\geq n^{6/7}$, the probability that $p^*\in M^*$ is at least $1/n^{1/7}$. 
We repeat the same process for $cn^{1/7}\log n$ times for a sufficiently large constant $c$, and return the largest clique we find. 
The probability that at least one selected point belongs to $M^*$ is at least 
$$
1- \left(1 - \frac{1}{n^{1/7}}\right)^{cn^{1/7}\log n} \geq 1 - e^{-c\log n}.
$$

Hence, with high probability, one of the $O(n^{1/7}\log n)$ randomly selected points is in $M^*$ and therefore the algorithm of Theorem~\ref{thm:convexgivenpoint} when applied to the point will compute a maximum clique of $G(P)$. As such, with high probability, our algorithm can find a maximum clique. The runtime of the algorithm is bounded by $O(n^2\log n\cdot n^{1/7}\log n)$ in the worst case, which is $O(n^{15/7} \log^2 n)$.

Combining the above, the time complexity of the overall algorithm is $O(n^{15/7+o(1)})$ in the worst case. The following theorem summarizes the result.

\begin{theorem}
Given a set $P$ of $n$ points in convex position in the plane, we can find a clique in the unit-disk graph $G(P)$ in $O(n^{15/7+o(1)})$ time, and with high probability the clique is a maximum clique of $G(P)$.
\end{theorem}

\bibliographystyle{plainurl}
\bibliography{refs}
\end{document}